%% file: main.tex
\documentclass[article]{elsarticle}
\usepackage{setspace}

\usepackage{graphicx}
\usepackage[]{algorithm2e}
\usepackage{algorithmicx}
\usepackage{url, lineno}
\usepackage{makeidx}
\usepackage{amssymb,graphics, amsmath, latexsym}
\usepackage{graphics, latexsym}
\usepackage[usenames]{color, epsfig}

\usepackage{changes}
\definechangesauthor[name={WW}, color=blue]{ww}

\journal{Theoretical Computer Science}

\newtheorem{theorem}{Theorem}

\newtheorem{lemma}[theorem]{Lemma}
\newtheorem{clm}[theorem]{Claim}

\newtheorem{definition}[theorem]{Definition}
\newproof{proof}{Proof}

\usepackage{fancyhdr}
\usepackage[yyyymmdd,hhmmss]{datetime}

\begin{document}
\begin{frontmatter}

\title{Distributed distance domination in graphs with no $K_{2,t}$-minor}

\author[2]{Andrzej Czygrinow\fnref{thanks}}
	\fntext[thanks]{Research supported in part by Simons Foundation Grant \# 521777.}
    \ead{aczygri@asu.edu}
    \author[1]{Micha\l\,Han\'{c}kowiak}
    \ead{mhanckow@amu.edu.pl}
    \author[1]{Marcin Witkowski}
    \ead{mw@amu.edu.pl}
    \address[2]{School of Mathematical and Statistical Sciences\\Arizona State University, Tempe, AZ,85287-1804, USA.}
    \address[1]{Faculty of Mathematics and Computer Science \\ Adam Mickiewicz University, Pozna\'n, Poland.}

\begin{abstract}
We prove that a simple distributed algorithm finds a constant approximation of an optimal distance-$k$ dominating set in graphs with no $K_{2,t}$-minor. The algorithm runs in a constant number of rounds. We further show how this procedure can be used to give a distributed algorithm which given $\epsilon>0$ and $k,t\in \mathbb{Z}^+$ finds in a graph $G=(V,E)$ with no $K_{2,t}$-minor a distance-$k$ dominating set of size at most $(1+\epsilon)$ of the optimum. The algorithm runs in $O(\log^*{|V|})$ rounds in the {\it Local} model. In particular, both algorithms work in outerplanar graphs.   
 \end{abstract}

\end{frontmatter}

\section{Introduction}
The minimum dominating set (MDS) problem is  notoriously difficult and yet extremely important because of its numerous applications. Recall that for a graph $G=(V,E)$, set $D\subseteq V$ is called a {\it dominating set} if every vertex in $V\setminus D$ has a neighbor in $D$. For general graphs $G$ even finding a  $C\log{n}$-approximation for some constant $C$ where $n$ is the order of $G$  is NP-hard \cite{RS}. At the same time, the problem becomes much more tractable when restricted to certain classes of graphs. In particular, assumptions about the sparsity of graphs, measured in various ways, can make the MDS problem easier to approximate. The situation is not much different in the distributed setting, where on one hand, only a $O(\log{\Delta})$-approximation for general graphs is known  \cite{KMW}, and on the other, the problem becomes significantly easier for special cases of graphs, like for example planar graphs \cite{LOW}.

In this paper, we shall consider an important generalization of the MDS problem, the distance-$k$ minimum dominating set problem, and we shall give fast deterministic distributed approximations in the {\it Local} model in the case the underlying network satisfies certain sparsity conditions.  

The term distance-$k$ dominating set was given by Henning et al. \cite{henning}. For a graph $G=(V,E)$ and $k\in \mathbb{Z}^+$, set $D\subseteq V$ is called a {\it distance-$k$ dominating set} if every vertex $v\in V$ is within distance $k$ of a vertex from $D$. In particular, $1$-distance dominating set is a dominating set. The problem has many applications in networking and other areas of computer science. Maybe the most natural applications of  distance-$k$ dominating sets arise when considering the problem of allocating centers in a network that can share resources with the remaining vertices of the graph when needed \cite{alloc}.
\subsection{Related Work}
In the distributed setting, the MDS problem has been extensively studied for many different classes of sparse networks. Lenzen et al. \cite{LOW} gave a constant-factor distributed approximation of a minimum dominating set that runs in a constant number of rounds in planar graphs in the {\it Local} model of computations. Using more careful analysis, Wawrzyniak \cite{ww-ipl} improved the approximation ratio and showed that this algorithm gives in fact a 52-approximation.   Amiri et al. \cite{amiri} showed that a small modification of the algorithm from \cite{LOW} also gives a constant-factor approximation of a minimum dominating set in graphs of the bounded genus, and even more generally in graphs with no $K_{3,t}$-minor for some constant $t$. In fact, a further generalization is given in \cite{CHWW} where the authors give a constant-time distributed algorithm for $K_t$-minor-free graphs.  In addition, using the methods from \cite{CHW}, it is possible to improve the approximation factor in these classes of graphs at the expense of the time complexity. Specifically, it can be proved that there is a distributed algorithm which given $\epsilon>0$ finds a $(1+\epsilon)$-approximation of a MDS in a graph $G=(V,E)$ that is $K_t$-minor-free in $O(\log^*{|V|})$ rounds.
For graphs of a constant arboricity, a much more general class of graphs, there is a randomized algorithm of Lenzen and Wattenhofer \cite{LW-arb} that finds a constant approximation in time which is $O(\log{|V|})$ rounds with high probability.
In addition, tight results are known for outerplanar graphs. Recently, using an analysis of a maximal counterexample, Bonamy, Cook, Groenland, and Wesolek \cite{B} manged to prove very tight bounds for the approximation ratio for MDS in the case of outerplanar graphs. Specifically they showed the following two facts.
\begin{itemize}
    \item There is a deterministic $5$-approximation of the MDS in outerplanar graphs.
    \item There is no $(5-\epsilon)$-approximation for outerplanar graphs for any $\epsilon>0$.
\end{itemize}
Very little is known about distributed algorithms for distance-$k$ dominating sets when $k>1$ as the problem becomes significantly different when $k$ increases making it impossible to adapt solutions for $k=1$. Amiri et al. \cite{amiri-ossona} gave a constant-factor approximation for the minimum distance-$k$ dominating set problem in graphs $G=(V,E)$ of bounded expansion in $O(\log{|V|})$ rounds (for a fixed $k$) in the more restrictive {\it Congest$_{BC}$} model. 

The main motivation for our work  comes from the recent paper by Amiri and Wiederhake \cite{AmiriWieder} who managed to provide a first constant approximation algorithm in a constant number of rounds for distance-$k$ domination in graphs of bounded expansion of high girth (i.e. graphs that are sparse and are trees locally).  In fact, we will use the very same procedure from \cite{AmiriWieder}, but give a different argument in the first part of the paper as we will examine a different class of graphs. Note that the girth assumption in \cite{AmiriWieder} was related to a previous work on lower bounds that were established for graphs of high girths and 
is absolutely critical to their analysis. It is this assumption that we will get rid of in the current paper (at the expense of dealing with graphs with no $K_{2,t}$-minor rather than a much more general class of graphs of bounded expansion). Therefore, our paper is a step towards constant time, constant approximation algorithms for minimum distance-$k$ domination in graphs of arbitrarily small girth for which constant time approximation are known.

It is worth mentioning that the problem for $k>1$ seems to be genuinely different than the classical MDS problem, especially in the realm of sparse graphs. For example, the probabilistic algorithm from \cite{LW-arb} is specific to the case $k=1$, and the methods from \cite{CHW} used to ameliorate the approximation ratio when a constant approximation is furnished are again applicable only to the regular distance-$1$ domination. 

\subsection{Summary of results}
We will work in the {\it Local} model of computations and assume throught the paper that $k\geq 2$. Although the first algorithm is identical to the algorithm from \cite{AmiriWieder} which works in the {\it Congest} model, the algorithm of Amiri et al.  exploits the fact that graphs are locally trees to allow for a {\it Congest} model implementation. Since the graphs considered in this paper can have many short cycles, the algorithm works only in the {\it Local} model. In addition, our algorithm for the $(1+\epsilon)$-factor approximation heavily relies on the assumptions of the {\it Local} model. In this model, vertices correspond to computational units, and computations are synchronized. In each round, a vertex can send, receive messages from its neighbors, and can perform individual computations. In addition, we assume that vertices have unique identifiers and denote the identifier of $v$ by $ID(v)$. 

Although our results are stated for graphs with no $K_{2,t}$-minor, an important subclass of this class  is outerplanar graphs that have no $K_{2,3}$-minor and no $K_4$-minor, that is graphs that admit a planar embedding in $\mathbb{R}^2$ such that all vertices lie on the boundary of the outer face. It would be possible to phrase the main result of the first part of the paper in a more general language of graphs of bounded expansion that are locally $K_{2,t}$-minor-free, but this would require additional terminology and the benefit seems quite minuscule. 

We will prove the following results. First, we will show that there is a distributed algorithm which finds a constant-approximation of a minimum distance-$k$ dominating set in graphs with no $K_{2,t}$-minor in a constant time which depends on $t$ and $k$ (Theorem \ref{const-approx}). Second, we will show that a suitable modification of methods from \cite{CHW} gives a $(1+\epsilon)$-approximation of the $k$-MDS problem in $O(\log^*{|V|})$ rounds in graphs $G=(V,E)$ that are $K_{2,t}$-minor-free (Theorem \ref{main-approx-thm}).

Finally we show that it is possible to find a $(1+\epsilon)$-factor approximation that runs in  $O(\log^*{|V|})$ rounds in  $K_t$-minor-free graphs $G=(V,E)$ of a constant maximum degree (Theorem \ref{const-deg-thm}).

The rest of the paper is structured as follows. In the next section, we shall fix some terminology and prove a fact about $K_{2,t}$-minor-free graphs that will be useful in the main part of the paper. Section \ref{sec-const} contains the analysis of the constant approximation algorithm and Section \ref{sec-eps} discusses the $(1+\epsilon)$-approximation.  

\section{Preliminaries}
Let $G=(V,E)$ and $H=(W,F)$ be graphs. We say that $G$ contains an $H$-minor if $H$ can be obtained from a subgraph of $G$ by a sequence of edge contractions. More formally, $H$ is a minor of $G$ if for some subgraph $G'=(V',E')$ of $G$ we can partition $V'$ into sets $V_1, \dots, V_l$ so that each $G'[V_i]$ is connected and the graph obtained from $G'$ by contracting every $V_i$ to a vertex is isomorphic to $H$. (Note that we discard all parallel edges or loops if they appear when contracting connected subgraphs.) We will be mainly interested in graphs $G$ that are $H$-minor-free (i.e. have no $H$-minor) for $H=K_{2,t}$ where $t\in \mathbb{Z}^+$ is a constant. Clearly, if a graph has no $K_{2,t}$-minor then it has no $K_{2,t+1}$-minor and  so assuming there is no $K_{2,t+1}$-minor is weaker than supposing no $K_{2,t}$-minor. Recall that if $G$ is planar, then $G$  has not $K_{3,3}$-minor and if it is outerplanar, then it has no $K_{2,3}$-minor. Consequently, our results apply to outerplanar graphs as $t$ can be a large but fixed positive integer.

A subdivision of a graph $H$, denoted $TH$, is obtained from $H$ by replacing its edges with internally disjoint paths of length at least one.

We will follow terminology from \cite{diestel} but will recall main concepts used throughout the paper. In particular, a path between two vertices does not contain a vertex more than once, a walk can contain repeated vertices or edges.

For two distinct vertices $u,v\in V$, a $u,v$-path is a path which ends in $u$ and $v$.  
We use $d_G(u,v)$ to denote the distance between $u$ and $v$ in $G$, that is, the length of a shortest $u,v$-path (allowing for $u=v$). 
For a subset $Q\subseteq V$, a $Q$-path is a path $P$ such that $V(P)\cap Q$ contains only the endpoints of $P$. In particular, every vertex of $Q$ is a trivial $Q$-path. For two disjoint sets $Q_1,Q_2$, a $Q_1, Q_2$-path is a path that has one endpoint in each of the $Q_i$'s and no other vertices in $Q_1\cup Q_2$. In the case $Q_1=\{u\}$, we will use $u,Q_2$-paths for $\{u\}, Q_2$-paths. We denote by $uPv$ a the subpath of $P$ between vertices $u$ and $v$.

For a vertex $v\in V$, $N(v), N[v]$ denote the neighborhood of $v$ and the closed neighborhood of $v$ respectively, that is $N[v]=\{v\} \cup N(v)$. In addition, for $l\in \mathbb{Z}^+$, let $N^l[v]$ denote the set of vertices within distance $l$ of $v$ and we set $N^l(v)=N^l[v]\setminus \{v\}$.
Similarly, for a set of vertices $X\subseteq V$, we let $N^l[X]=\bigcup_{v\in X}N^l[v]$.

Finally, we will use $\gamma_k(G)$ to denote the size of a smallest distance-$k$ dominating set in $G$.

Graphs with no $K_{2,t}$-minor are sparse. To be precise, we have the following fact \cite{CRS}.
\begin{lemma}\label{crs-lem}
Let $t\geq 2$ and let $H$ be a graph of order at least one with no $K_{2,t}$-minor, then $|E(H)|\leq \frac{1}{2}(t+1)(|V(H)|-1).$
\end{lemma}
Although we are not going to attempt to optimize the constants and do not really need the full power of the previous lemma, we will use it to obtain a bound for the number of vertices on $Q$-paths of length at most $h$.

\begin{lemma}\label{main-lem-estimate}
Let $t\geq 2$, $h\in \mathbb{N}$ and let $H=(V,E)$ be a graph with no $K_{2,t}$-minor. Let $Q\subseteq V$ and let $Q_h$ denote the set of vertices which are on $Q$-paths in $H$ of length at most $h$. Then $|Q_h|\leq \alpha_{h,t} |Q|$ for some $\alpha_{h,t}$ that depends on $h$ and $t$ only. 
\end{lemma}
\begin{proof} We will induct on $h$. If $h=0$ then $Q_0=Q$, and $\alpha_{0,t} = 1$.  For the inductive step, let $\mathcal{P}$ be a maximal set of $Q$-paths of length at most $h$ which are internally disjoint. For $u,v\in Q$ let $P_{u,v}$ denote $u,v$-paths in $\mathcal{P}$ and note that $|P_{u,v}|\leq t+1$ because otherwise, graph $H$ has a subdivision of $K_{2,t}$ (as there can be only one edge $uv$, and every other path has length at least two), and so a $K_{2,t}$-minor. Contract paths from $P_{u,v}$ to edge $uv$ and apply Lemma \ref{crs-lem} to conclude that the number of edges in the contracted graph  is less than $(t+1)|Q|/2$. Consequently, since each path has length at most $h$, the number of vertices on paths from $\mathcal{P}$ is less than $(t+1)^2(h+1)|Q|/2$. Let $Q'$ denote the set of vertices on paths from $\mathcal{P}$ (including paths of length $0$, or $1$, so notice that it does contain every vertex from set $Q$). 

If $S$ is a $Q$-path of length at most $h$ which does not belong to $\mathcal{P}$, then $S$ contains an internal vertex of a path from $\mathcal{P}$ and so, all vertices on $S$ that have not been already counted belong to $Q'$-paths of length at most $h-1$ (see Figure~\ref{fig1}). Thus, by induction, the number of vertices on such paths is at most $\alpha_{h-1, t}|Q'|$. Consequently, by (\ref{eqQ}), the number of vertices on all paths is at most 
\begin{equation}\label{eqQ2}
\begin{split}
|Q_h| & \leq (1+\alpha_{h-1,t})|Q'| \\
& <(1+\alpha_{h-1,t})(t+1)^2(h+1)|Q|/2 \\
& \leq  (t+1)^2(h+1)\alpha_{h-1,t}|Q|,
\end{split}
\end{equation}
which gives us rough estimate on $\alpha_{h,t} \leq (t+1)^{2h}(h+1)!$, and so
\end{proof}
\begin{equation}\label{eqQ}
    |Q'|\leq (t+1)^2(h+1)|Q|/2.
\end{equation}
\begin{figure}
\begin{center}
 \scalebox{0.7}{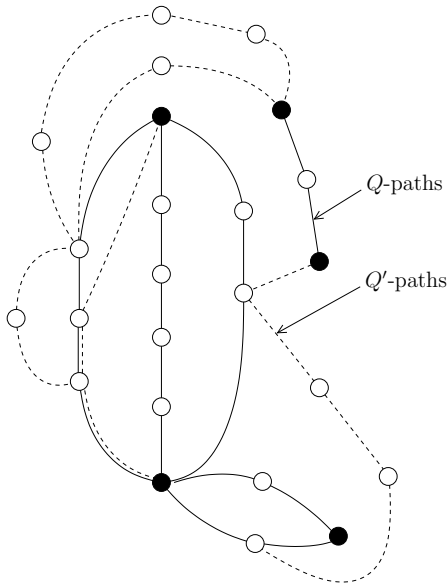}
\caption{\normalsize Example for Lemma 2. The set $Q$ consists of the five black vertices, $t = 3$, $h = 5$. Every $Q$-path not in $P$ has its vertices covered by $P$ and $Q'$-path.} 
\label{fig1}
\end{center}
\end{figure}
\section{Constant factor approximation}\label{sec-const}
In this section, we will show that the simple algorithm (Algorithm 1 from \cite{AmiriWieder}) finds a distance-$k$ dominating set of size  $O(\gamma_k(G))$ in graphs $G$ with no $K_{2,t}$-minor. 

We work in the {\it Local} model, and as a result we may assume that each connected component of $G$ has diameter at least $4k$. Indeed, if component $C$ has diameter less than $4k$, then a simple $O(k)$-round algorithm can test that $C$ is a component and will compute an optimal distance-$k$ dominating set in $C$. In particular, $t\geq 2$.
In the procedure from \cite{AmiriWieder},  every vertex $v\in V$ selects vertex $w\in N^k[v]$ with $|N^k[w]|$  maximum and resolves ties using $ID(w)$. 
More formally, the algorithm can be described as follows:\\
\begin{algorithm}[H]
 \KwData{Graph $G=(V,E)$}
 \KwResult{Set $D$ }
 \label{alg1}
 \caption{{\sc DomSet}}
\begin{enumerate}
    \item For every $v\in V$, in parallel, find $q_v=|N^k[v]|.$
    \item For every $v\in V$ let $w:=w_v$ be the vertex in $N^k[v]$ such that 
    \begin{itemize}\item $q_w$ is maximum,
    \item and subject to this, $ID(w)$ is maximum.
    \end{itemize}
    \item Return $D:=\bigcup \{w_v\}$.
\end{enumerate}
\end{algorithm}
The algorithm clearly runs in $O(k)$ rounds and outputs a distance-$k$ dominating set. The only difficulty is to show that it indeed finds a distance-$k$ dominating set $D$ such that $|D|=O(\gamma_k(G))$, which we will do in the remainder of this section.

In our analysis we may assume that $G$ is connected because the same argument can be applied to each connected component.

Let $M$ be an optimal distance-$k$ dominating set in $G=(V,E)$. Create Voronoi cells (also called clusters) centered at $M = \{u_1,...,u_m\}$ with a vertex $v$ joining cell $C_{i}$ if $d_G(u_i,v)$ is the smallest and with ties resolved by selecting $u_i$ with the maximum ID. This gives a set of cells $\mathcal{C}=\{C_1, \dots, C_m\}$ such that each $G[C_i]$ is connected. For a cell $C\in \mathcal{C}$ let $v_C$ be the vertex in $C$ such that $d_G(v_C,w)\leq k$ for every $w\in C$, and subject to this, $|N^k_G(v_C)|$ is maximum, and subject to that, $ID(v_C)$ is maximum. If $C=C_i$, then $v_C$ might be the same as $u_i$ or it can be a different vertex but $u_i$ is always an option.

\begin{definition}
Let $C\in \mathcal{C}$. A vertex $v\in C$ is called a border vertex if  $v$ has a neighbor in $V\setminus C$.
\end{definition}
Let $C^*$ denote the set of border vertices in $C$.
We have the following simple observation.
\begin{lemma}\label{simple-lem}
Let $G=(V,E)$ be a connected graph of diameter at least $4k-1$ and let $C$ be a Voronoi cell. If $y,y^*\in C$ are such that $N^k[y]\subseteq N^k[y^*]$ and $d_G(y^*,w)<d_G(y,w)\leq k$ for every $w\in C^*$, then $N^k[y^*]\setminus N^k[y]\neq \emptyset.$
\end{lemma}
\begin{proof} We have $N^{k}[C]\setminus N^{k-1}[C]\neq \emptyset$ because otherwise $N^{2k-1}[v_C]=V$. Let $v \in N^{k}[C]\setminus N^{k-1}[C]$ and let $Q$ be a shortest $v,C$-path in $G$. Let $\{w\}=Q\cap C$. Then for any $i=1,\dots, k$ there is a vertex $x_i$ on $Q$ such that $d_G(x_i, C) =d_G(x_i,w) = i$. Since $d_G(y^*,w)< d_G(y,w)\leq k$ there is a vertex on $Q$ at distance $k$ from $y^*$ and at distance $k+1$ from $y$. \end{proof}

We will now prove our main lemma which shows that the number of vertices from $C$ that can be added by the Algorithm \ref{alg1} is $O(|C^*|)$. The idea is to show that the number of potential choices for vertices selected by {\sc DomSet} can be bounded from above by repeated application of Lemma \ref{main-lem-estimate}. 
\begin{definition}
 Let $U_0:=C^*\cup \{v_C\}$ and for $i>0$ let $U_i$ be the set of vertices on $U_{i-1}$-paths of length at most $3k$.
\end{definition}
Note that we have $U_{i-1}\subseteq U_{i}$ because of the trivial paths.
\begin{lemma}
\label{lemmaU_k}
Output of the {\sc DomSet} contains only vertices of $U_k$.
\end{lemma}
\begin{proof}
Assume towards contradiction that for some $C$ there is a vertex  $y\in C$ selected by {\sc DomSet} such that $y \notin U_k$. 
We have $d_{G}(y,v_C)\leq k$. Fix one $y,v_C$-path of the shortest length in $G$ and call it $P$. Let $l \leq k$ denote the length of $P$. We will now select a vertex which is closest to $y$ on $P$ and belongs to the sets $U_{l-i}$ for some $i\geq 0$. 
Let $i$ be the smallest non-negative integer such that there is a vertex $y^*$ in $U_{l-i}\cap V(P)$ which satisfies $d_P(y,y^*)\leq i$. We choose a vertex $y^*$ for which $d_P(y,y^*)$ is the smallest. Since $v_C\in U_0$ and $d_P(v_C,y)=l$, we can deduce that such $y^*$ always exists and $i\leq l \leq k$. 

\begin{figure}
\label{fig2}
\begin{center}
 \scalebox{0.8}{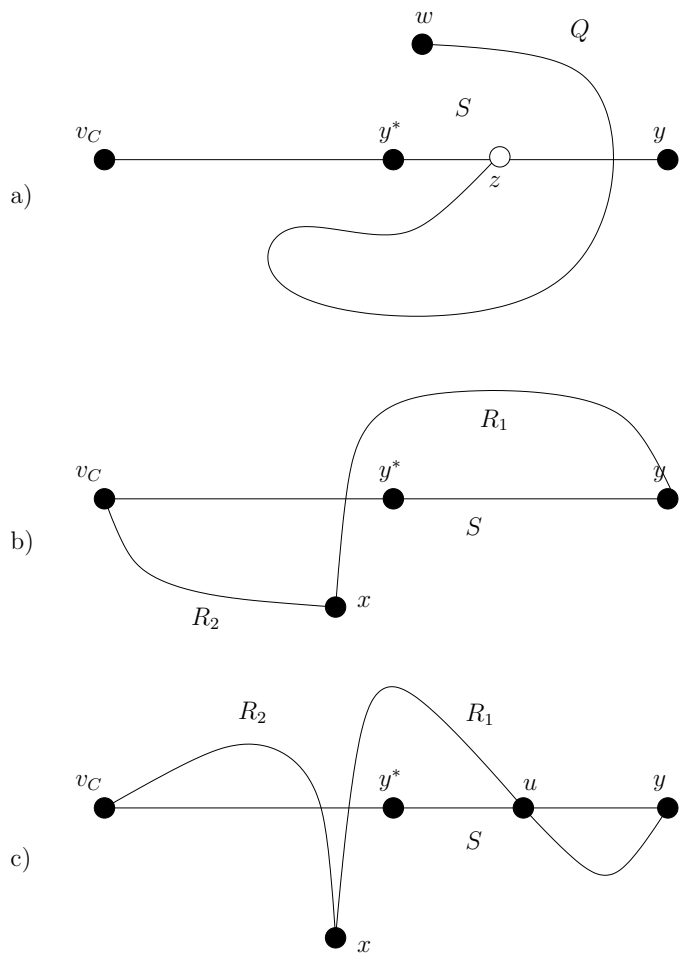}
\caption{\normalsize Example of Lemma \ref{lemmaU_k}. a) Path $S$ connecting $y$ and $y^*$. b) Path $R_1$ and $R_2$ having no intersection with $S$. c) Path $R_1$ having an intersection with $S$ in vertex $u$.}
\end{center}
\end{figure}

Note that for the case $i=0$, we have $y=y^*\in U_l\subseteq U_k$. Thus,  consider cases when $y^*\in U_{l-i}$ where $i\geq 1$, and let $S=yPy^*$.
We will now analyze possibilities of the placement of $y^*$ in relation to $y$ (see Figure \ref{fig2}) and prove the following three claims.
\begin{enumerate}
\item \begin{clm}\label{cl1}
If $w\in C^*$ and $d_{G}(y,w)\leq k$ then $d_G(y^*,w) < d_G(y,w)$.
\end{clm}
\begin{proof}
Let $Q$ be a shortest $y,w$-path of length at most $k$ and assume towards contradiction that $Q$ does not contain $y^*$  (otherwise the condition is trivially true). If $Q\cap S=\{y\}$, then $y\in U_{(l-i)+1}\subseteq U_k$ because the length of $Q\cup S$ is at most $2k$ and both $w,y^*\in U_{l-i}$. 

If $Q\cap S\neq\{y\}$ then there is another vertex $z\in Q\cap S$ with $z\neq y$. If $d_G(z,y)>d_G(z,y^*)$ then $d_G(y^*,w) < d_G(y,w)$. Thus assume  $d_G(z,y)\leq d_G(z,y^*)$ and $zSy^*$ does not contain any vertices of $Q$. By the choice of path $P$ we have that $d_P(z,y)<d_P(y^*,y)$. Then, by definition of $y^*$, we have $z\neq w$, but, as above, since $zSy^* \cap zQw=\{z\}$, we have $z\in U_{(l-i)+1}$, which gives a contradiction with the choice of $y^*$.

Thus for every $w\in C^*$ such that $d_G(y, w)\leq k$ every  shortest $y,w$-path $Q$ in $G$ contains $y^*$, and so $d_{G}(y^*,w)< d_G(y,w)$. 
\end{proof}
\item \begin{clm}\label{cl2}
$N^k[y]\subseteq N^k[y^*]$
\end{clm}
\begin{proof}
Let $x\in N^k[y]$, then we have $d_G(x,y)\leq k$ and we show that $d_G(x, y^*)\leq k$.  Clearly $d_G(y^*,y)\leq k$ and so we may assume that $x\neq y$. In addition, if $x\notin C$ then by the previous argument $d_G(x,y^*)<d_G(x,y)$, because $d_G(w, y^*)<d_G(w,y)$ for every $w\in C^*$. Therefore, we may assume that $x\in C$. 

Let $R_1$ be a $y,x$-path of length at most $k$. Since $x$ is in $C$, there is a $v_C,x$-path $R_2$ of length at most $k$. If $R_1\cap S=\{y\}$ and $R_2\cap S=\emptyset$, then the walk $R_1\cup R_2$ contains a $v_C,y$-path $R$ such that $R\cap S=\{y\}$. The length of $R$ is at most $2k$ and the length of $S$ is at most $k$. Consequently $y\in U_{k}$, which is a contradiction.

In addition, we either have $y^*\in R_2$ or $R_2\cap S=\emptyset$ (again if $y^*\notin R_2$ and there is some $u\in R_2\cap S$ that is closest to $y^*$ on $P$, then $u\in U_{(l-i)+1}$, contradicting the choice of $y^*$). In the former case $d_G(x,y^*)\leq d_{R_2}(x,v_C)\leq k$ and so we may assume that $R_2\cap S=\emptyset$. By the previous discussion, the case that is left to analyze is when $R_1\cap S$ contains some vertex $u\neq y$. If $R_1\cap S$ contains $y^*$, then $d_G(y^*,x)\leq k$. Otherwise selecting $u$ to be closest to $y^*$ gives a $u,x$-path  of length at most $k$ which in connection with the $v_C,x$-path $R_2$ gives a $u,v_C$-walk of length at most $2k$. This walk intersects $S$ in $u$ only. Consequently $u\in U_{(l-i)+1}$ because $y^*\in U_{l-i}$ and $u$ belongs to a $y^*,v_C$-path of length at most $3k$. This again contradicts the choice of $y^*$. 
Therefore, $N^k[y]\subseteq N^k[y^*]$. 
\end{proof}
\item \begin{clm}\label{cl3}
$N^k[y]$ is a proper subset of $N^k[y^*]$.
\end{clm}
\begin{proof}
By Claim \ref{cl2}, $N^k[y]\subseteq N^k[y^*]$. Suppose there exists $w\in C^*$ such that $d_G(y,w)\geq k+1$ and let $Q$ be a shortest $v_C,w$-path in $G$. Note that  the length of $Q$ is at most $k$ and so $d_G(y,w)\leq 2k$ as witnessed by $P\cup Q$.
If there is a $y,w$-path $R$ of length at most $2k$ such that $S\cap R$ contains $z\neq y^*$, and $y^*\notin R$ then choosing $z$ closest to $y^*$ on $S\cap R$ gives $z\in U_{(l-i)+1}$ contradicting the choice of $y^*$.  Thus every such path $R$ contains $y^*$ because otherwise $S\cap R=\{y\}$ implying $y\in U_k$. Considering a $y,w$-path of  the shortest length gives a vertex $z$ such that $d_G(y^*,z)\leq k$  and $d_G(y,z)>k$. Thus $z\in N^k[y^*]\setminus N^k[y]$.
If for every $w\in C^*$ $d_G(y,w)\leq k$ then by Claim \ref{cl1} for every $w\in C^*$ $d_G(y^*,w)< d_G(y,w)\leq k$ and so, by Lemma \ref{simple-lem},  $N^k[y^*]$ is a proper subset of $N^k[y]$.
\end{proof}
\end{enumerate}
Therefore, by Claim \ref{cl3}, $N^k[y]$ is a proper subset of $N^k[y^*]$ and so {\sc DomSet} chooses $y^*\in U_k$. \end{proof}

\begin{lemma}\label{main-const-approx}
For every $t,k\in \mathbb{Z}^+$ there is $\beta_{k,t}$ such that the number of vertices in $C$ selected by {\sc DomSet} is at most $\beta_{k,t}|C^*|.$
\end{lemma}
\begin{proof} Since output of {\sc DomSet} is a subset of $U_k$ and in view of Lemma \ref{main-lem-estimate}, $|U_i|\leq \alpha_{3k,t}|U_{i-1}|$.
\end{proof}

\begin{definition}
Let $V^*=\bigcup_{C\in \mathcal{C}} C^*$.
\end{definition}
Using a few relatively easy lemmas, we can conclude the analysis of {\sc DomSet}.
\begin{lemma}\label{lem-simple-W}
Let $C\subseteq V$ be such that $G[C]$ is connected, and such that for some vertex $v_C$, $d_{G[C]}(v_C, w)\leq k$ for every $w\in C$.
Let  $W\subseteq C$ be a set which satisfies $|W|> ks^{k}$.  Then $G[C]$ contains a subdivision of $K_{1,s}$ with all leaf vertices in $W$. 
\end{lemma}
\begin{proof}  Let $T$ denote a spanning BFS tree in $G[C]$ rooted at $v_C$ and let $W_i$ denote the set of vertices in $W$ that are at distance $i$ from $v_C$ in $T$. We have $\sum_{i=0}^k|W_i| =|W|$, and so there is an $i$ such that $|W_i|\geq s^{k}$. For $w\in W_i$, let $P_w$ denote the path $wTv_C$ (path from $w$ to $v_C$ using the edges of the spanning tree $T$), and let $T'$ be the union $\bigcup_{w\in W_i} P_w$. Then $T'$ is a tree with leaves in $W_i$ and for every $w\in W_i$, $d_{T'}(w,v_C)\leq k$. If there is a vertex $z\in T'$ such that $deg_{T'}(z)\geq s$, then $T'$, and so $G[C]$, contains a subdivision of $K_{1,s}$ with all leaf vertices in $W$. Otherwise, the number of vertices in $W_i$ is less
than $s^{k}$. \end{proof}
\begin{lemma}\label{lem-two-clusters-edges}
Let $C,C'$ be two Voronoi cells as in Lemma \ref{lem-simple-W}. Then the number of edges between $C$ and $C'$ is at most $k^2t^{2kt^k}$.
\end{lemma}
\begin{proof} If there is a vertex $z\in C$ which has more than $kt^{k}$ neighbors in $C'$, then, by Lemma \ref{lem-simple-W} applied to $W=N(z)\cap C'$, $G$ contains a subdivision of $K_{2,t}$. If there is a matching $Q$ between $C$ and $C'$ of size larger than $kt^{kt^k+1}$, then we can apply Lemma \ref{lem-simple-W} twice. Apply it first with $s=kt^k+1$  and $W= V(Q)\cap C'$ to get a subdivision of $K_{1,s}$, $T$, in $C'$. Then apply it again with $s=t$  and $W\subseteq V(Q)\cap C$ which contains vertices matched by $Q$ with the leaves of $T$ to get a subdivision of $K_{1,t}$ in $C$. Therefore the number of edges between $C$ and $C'$ is at most $k^2t^{2kt^k}.$ \end{proof}

Finally, we have with the following observation (we recall that $M$ is an optimal distance-$k$ dominating set in $G$).

\begin{lemma}\label{bound-lemma}
$|V^*|\leq  k^2t^{2kt^k}(t+1)|M|.$
\end{lemma}
\begin{proof} Contracting each $C\in \mathcal{C}$ to a vertex gives a minor of $G$ which by Lemma \ref{crs-lem} has at most $(t+1)|M|/2$ edges. By Lemma \ref{lem-two-clusters-edges}, the number of edges with endpoints in two different Voronoi cells is at most $k^2t^{2kt^k}(t+1)|M|/2$. Consequently, the  number of vertices that belong to these edges is at most $k^2t^{2kt^k}(t+1)|M|$. \end{proof}

We will now combine the previous facts to prove the main result of this section.

\begin{theorem}\label{const-approx}
Let $t, k\in Z^+$. Then there exists $\delta=\delta(t,k)$ such that given a connected graph $G$ with no $K_{2,t}$-minor and such that $diam(G)\geq 4k$, algorithm {\sc DomSet} finds in $O(k)$ rounds a distance-$k$ dominating set $D$ in $G$ such that $|D|\leq \delta \cdot \gamma_k(G).$   
\end{theorem}
\begin{proof}
Let $M$ be an optimal distance-$k$ dominating set in $G$.
By Lemma \ref{main-const-approx}, the set $D$ obtained by {\sc DomSet} satisfies $|D|\leq \beta_{k,t}|V^*|$ for some $\beta_{k,t}$. In view of Lemma \ref{bound-lemma}, $|V^*|\leq k^2t^{2kt^k}(t+1)|M|$ and so, $|D|\leq  \delta |M|$ for $\delta =\beta_{k,t} k^2t^{2kt^k}(t+1)$. \end{proof}

\section{The $(1+\epsilon)$-factor approximation}\label{sec-eps}
In this section, we will first give a distributed $(1+\epsilon)$-factor approximation of an optimal $k$-MDS in the case when $G$ is $K_{2,t}$-minor-free. This algorithm runs in $O(\log^*{|V|})$ rounds in the {\it Local} model.
As noted in the introduction, adapting methods from \cite{CHW} is not automatic. However, there are some instances when this can be accomplished with relatively little effort.  In the second part of this section, we give one example of such a situation when a graph $G$ is $K_t$-minor-free and satisfies  $\gamma_k(G)\leq C \gamma_1(G)$ for some constant $C$. For example, $K_t$-minor-free graphs of a bounded maximum degree satisfy this condition.
\subsection{Graphs with no $K_{2,t}$-minor}

Let $H=(W,F)$ be a graph and let $P=(W_1, \dots, W_l)$ be an ordered partition of $W$. We define $\partial(P)$ to be the set of vertices $v\in W$ such that $v \in W_i$ and $N(v)\cap W_j\neq \emptyset$ for some $i\neq j$.
We have the following theorem which can be proved by applying methods from \cite{CHW} and is a special case of the corresponding theorem in \cite{CHWW1}.
\begin{theorem}\label{waw-thm} Let $s\in \mathbb{Z}^+$ and let $\epsilon>0$. There exists $L$ such that the following holds. Let $H=(W,E)$ be a graph on $n$ vertices with no $K_{s}$-minor. There is a distributed algorithm which finds a partition $P=(W_1, \dots, W_l)$ such that: 
\begin{itemize}
    \item For every $i$, $H[W_i]$ has diameter $O(L)$ and
    \item $|\partial(P)|\leq \epsilon |W|$.
    \end{itemize}
The algorithm runs in $L\log^*{n}$ rounds.
\end{theorem}
We will use the algorithm from  Theorem \ref{waw-thm} to improve the approximation ratio of the algorithm from the previous section. Although the general idea is the same as in \cite{CHW}, there are a few changes in the analysis that must be made to account for the fact that we are dealing with a distance-$k$ dominating set with $k\geq 2$. In particular, the assumption that there is no $K_{2,t}$-minor (rather than a more relaxed assumption that there is no $K_{s}$-minor for $s\geq t+2$) will play a critical role in the analysis via Lemma \ref{two-cluster-lem}.
 
Let $\alpha \in (0,1)$ be given and let $D$ be the set obtained by {\sc DomSet}. Then, by Theorem \ref{const-approx}, \begin{equation}\label{eq-C}|D|\leq \delta \cdot \gamma_k(G)\end{equation} for some $\delta$ that depends on $t,k$ only.
Consider Voronoi cells with centers in vertices from $D$, that is, Voronoi cells $C_v$ for $v\in D$ with $w$ joining $C_v$ if $d_G(v,w)$ is minimum over all $v\in D$ and ties resolved by selecting $v$ with maximum $ID(v).$ Let $H=(W,F)$ be obtained from $G$ by contracting each $C_v$ to a vertex. 
Set $\epsilon :=\frac{\alpha}{2\delta k^2t^{2kt^k}}$ and let $P = (W_1, \dots, W_l)$ be the partition of $W$ from Theorem \ref{waw-thm}. We have \begin{equation}\label{partial1}|\partial(P)|\leq \epsilon |W| =\epsilon |D|.\end{equation}
Partition $P$ yields partition $P'=(V_1, \dots, V_l)$ of $V(G)$ by setting $V_i:=\bigcup_{u\in W_i} C_u$.
\begin{lemma}\label{two-cluster-lem}
Let $u,w\in V(H)$. Then the number of edges in $G$ between $C_u$ and $C_v$ satisfies $|E_{G}(C_u, C_v)|\leq k^2 t^{2kt^k}.$
\end{lemma}
\begin{proof} By construction every vertex $w\in C_u$ is within distance $k$ of $u$. Consequently, by Lemma \ref{lem-two-clusters-edges}, $|E_G(C_u,C_v)|\leq k^2t^{2kt^k}.$ \end{proof}

Now combining Lemma \ref{two-cluster-lem} and (\ref{partial1}) we have that $\epsilon |D|$ Voronoi cells can have $k^2t^{2kt^k}$ edges between them, each connecting a pair of vertices. Thus
\begin{equation}\label{partial2}
|\partial(P')|\leq 2\epsilon k^2t^{2kt^k}|D|.
\end{equation}
Informally speaking, we use Algorithm 1 to find a seed dominating set. We define Voronoi cells and construct groups of Voronoi cells using Theorem \ref{waw-thm}, and solve the subgraphs inside these groups optimally. Specifically, we consider the following procedure.\\
\begin{algorithm}[H]
 \KwData{Graph $G$ with no $K_{2,t}$-minor, $k\in Z^+$, $0<\alpha<1$}
 \KwResult{Set $Q$}
\caption{{\sc $k$-DomSet Approximation}}
\begin{enumerate}
    \item Find $D$ using {\sc DomSet}.
    \item Construct graph $H$ as above and set $\epsilon:=\frac{\alpha}{2\delta k^2t^{2kt^k}}$.
    \item Use the algorithm from Theorem \ref{waw-thm} to find $P$. Let $Q:= \partial(P')$.
    \item For every $i=1, \dots, l$ find a set $Q_i$ in $G[V_i]$ such that $|Q_i|$ is the smallest and $Q_i \cup (\partial(P')\cap V_i)$ distance-$k$ dominates $V_i$ in $G$.
    \item Return $Q:=Q\cup \bigcup_i Q_i.$
\end{enumerate}
\end{algorithm}
Using the above discussion we can now prove the main theorem.
\begin{theorem}\label{main-approx-thm}
Let $\alpha\in (0,1)$ and let $t,k\in Z^+$. Given a connected graph $G=(V,E)$ with no $K_{2,t}$-minor of diameter at least $4k$, procedure {\sc $k$-DomSet Approximation} finds in $O(\log^*{|V|})$ rounds set $Q$ such that $|Q|\leq (1+\alpha)\gamma_k(G).$   
\end{theorem}
\begin{proof} The algorithm runs in $O(kL \log^*{|V|})=O(\log^*{|V|})$ rounds (where $L$ is the constant from Theorem \ref{waw-thm}) because $diam(G[V_i])=O(kL)$ by Theorem \ref{waw-thm} and the construction, and so step 4 requires $O(kL)$ rounds in the {\it Local} model.

Let $M$ be an optimal distance-$k$ dominating set in $G$, $i\in \{1, \dots, l\}$, and $M_i:=M\cap V_i$.
Let $V_i^O$ denote the set of vertices $w\in V_i$ such that $d_G(w, \partial(P')\cap V_i)\leq k$, and let $V_i^I:=V_i\setminus V_i^O$. 
Clearly every vertex from $V_i^O$ is distance-$k$ dominated by $\partial(P')\cap V_i$. In addition, if $w\in V_i^I$, then $w$ must be distance-$k$ dominated by a vertex from $M_i$. Thus by step 4 of {\sc $k$-DomSet Approximation}, we have $|Q_i|\leq |M_i|$, and $|Q|\leq |\partial(P')|+\sum_{i=1}^l |M_i|$, 
which in view of (\ref{partial2}) gives $|Q|\leq 2\epsilon k^2t^{2kt^k} |D|+|M|.$
Further by (\ref{eq-C}) and the definition of $\epsilon$ we have
$$
|Q|\leq (2\epsilon k^2t^{2kt^k}\delta+1)\gamma_{k}(G)=(1+\alpha)\gamma_k(G).
$$

\end{proof}
\subsection{$K_t$-minor free graphs of a constant maximum degree}
In this last short section, we show a simple method to find a $(1+\epsilon)$-factor approximation of the minimum distance-$k$ dominating set if $G$ is $K_t$-minor-free and the maximum degree of $G$ satisfies $\Delta(G)\leq L$ for some $L$ independent of $G$. 

In fact, we will give an algorithm for a somewhat more general class of $K_t$-minor-free graphs that we call $(C, \gamma_k)$-bounded. 

Note that obviously, for every graph $G$ and every $i\in Z^+$, we have $\gamma_{i}(G)\geq \gamma_{i+1}(G)$.

Fix $k\in Z^+$. We say that a graph $G$ is $(C,\gamma_k)$-bounded if  $\gamma_1(G)\leq C \gamma_k(G)$. For example, if for some $L\geq 3$, graph $G=(V,E)$ is such that $\Delta(G)\leq L$, then $\gamma_{k}(G)> |V|/L(L-1)^k\geq  \gamma_1(G)/L(L-1)^k$ and so $G$ is $(C,\gamma_k)$-bounded with $C=L(L-1)^k$.

The algorithm is very simple and we will only outline the main idea.
Fix $C,k, t$ which are known to the algorithm and let $G$ be a graph that is $K_t$-minor-free and $(C,\gamma_k)$-bounded. Find a constant approximation $S$ of a minimum dominating set in $G$ (i.e. distance-$k$ dominating set  with $k=1$) by using the algorithm from \cite{CHWW}. Partition $V(G)$ into $\{S_v|v\in S\} $ by setting $S_v=\{v\}$ and adding $u$ to $S_v$  if $uv\in E$ and $v$ has the maximum ID over vertices from $S$. Construct $H$ by contracting each $S_v$ to a vertex. Set $\epsilon$ appropriately and use Theorem \ref{waw-thm} to find a partition $(W_1,\dots, W_l)$ of $V(H)$. If a vertex in $S_v\in W_i$ has a neighbor in $W_j$ for some $j\neq i$ in $G$ (border set), then add the center of $S_v$, namely $v$, to $D$. Finally, for each $i=1 ,\dots, l$ find an optimal set $D_i$ in $G$ such that $D_i\cup D$ distance-$k$ dominates $V_i:=\bigcup_{S_u\in W_i} S_u$ in $G$. 
\begin{theorem}\label{const-deg-thm}
Let $C,k,t\in Z^+$ and let $0<\alpha <1$. There is a distributed algorithm which given a $K_t$-minor-free graph $G=(V,E)$ which is $(C,\gamma_k)$-bounded, finds in $O(\log^*{|V|})$ rounds a set $D\subseteq V(G)$ such that $|D|\leq (1+\alpha)\gamma_k(G).$  
\end{theorem}
\begin{proof} (Sketch) The argument is analogous to the proof of Theorem \ref{main-approx-thm}. In particular, it is easy to see that the number of vertices added to $D$ from border sets $S_v$ is  $O(\epsilon \gamma_1(G))$ which can be made smaller than $\alpha \gamma_k(G)$ using appropriately defined $\epsilon$ and the fact that $G$ is $(C,\gamma_k)$-bounded. Now assume $w \in V_i$ and $w$ is distance-$k$ dominated by a vertex $u\in V_j$ for some $j\neq i$. Then a shortest $u,w$-path $P$ in $G$ contains a vertex $x$ from some border set $S_v\in W_i$. Then however $vxPu$ has length which is less than or equal to the length of $P$ and so the center of $S_v$,  which is added to $D$, distance-$k$ dominates $w$. Finally, vertices $w\in V_i$ that are not distance-$k$ dominated by vertices from other Voronoi cells than $V_i$ are distance-$k$ dominated by $D_i$ and $|D|=\sum_{i}|D_i|\leq \gamma_k(G)$. \end{proof} 

\section{Conclusions}
We finish  with a short summary. In this paper we considered the distance-$k$ dominating set problem in a special class of graphs and showed three facts. 
\begin{itemize}
    \item There is a (simple) distributed (LOCAL) constant-factor approximation of an optimal distance-$k$ dominating set in graphs $G$ that have no $K_{2,t}$-minor. The algorithm runs in $O_t(k)$ rounds.
    \item There is a distributed (LOCAL) algorithm which given $\epsilon>0$ finds in a $K_{2,t}$-minor free graph $G$ of order $n$ a distance-$k$ dominating set of size at most $(1+\epsilon)\gamma_k(G).$ The algorithm runs in $O_{\epsilon, k, t}(\log^*{n})$ rounds. 
    \item There is  a distributed (LOCAL) algorithm which given $\epsilon>0$ finds in a $K_{t}$-minor free graph $G$ of a constant maximum degree and order $n$ and a distance-$k$ dominating set of size at most $(1+\epsilon)\gamma_k(G).$ The algorithm runs in $O_{\epsilon, k, t,\Delta(G)}(\log^*{n})$ rounds. 
    \end{itemize}
    The proofs of the first two statements critically rely on the fact that graph $G$ is $K_{2,t}$-minor free and the third one applies only to very restrictive class of graphs, the class of $K_t$-minor-free graphs of a constant maximum degree. It would be interesting to see if similar facts can be obtained for graphs with no $K_{3,t}$-minor, or in general for graphs which are $K_t$-minor free.

\end{document}

%% file: picture1.pdf_t
\begin{picture}(0,0)%
\includegraphics{picture1.pdf}%
\end{picture}%
\setlength{\unitlength}{2486sp}%
\begingroup\makeatletter\ifx\SetFigFont\undefined%
\gdef\SetFigFont#1#2#3#4#5{%
  \reset@font\fontsize{#1}{#2pt}%
  \fontfamily{#3}\fontseries{#4}\fontshape{#5}%
  \selectfont}%
\fi\endgroup%
\begin{picture}(5570,8225)(2836,-7161)
\put(7947,-2941){\makebox(0,0)[lb]{\smash{{\SetFigFont{12}{14.4}{\rmdefault}{\mddefault}{\updefault}{\color[rgb]{0,0,0}$Q'$-paths}%
}}}}
\put(7962,-1561){\makebox(0,0)[lb]{\smash{{\SetFigFont{12}{14.4}{\rmdefault}{\mddefault}{\updefault}{\color[rgb]{0,0,0}$Q$-paths}%
}}}}
\end{picture}%

%% file: picture2.pdf_t
\begin{picture}(0,0)%
\includegraphics{picture2.pdf}%
\end{picture}%
\setlength{\unitlength}{2486sp}%
\begingroup\makeatletter\ifx\SetFigFont\undefined%
\gdef\SetFigFont#1#2#3#4#5{%
  \reset@font\fontsize{#1}{#2pt}%
  \fontfamily{#3}\fontseries{#4}\fontshape{#5}%
  \selectfont}%
\fi\endgroup%
\begin{picture}(8340,11790)(1696,-10456)
\put(6301,-4471){\makebox(0,0)[lb]{\smash{{\SetFigFont{12}{14.4}{\rmdefault}{\mddefault}{\updefault}{\color[rgb]{0,0,0}$y^*$}%
}}}}
\put(6031,-6091){\makebox(0,0)[lb]{\smash{{\SetFigFont{12}{14.4}{\rmdefault}{\mddefault}{\updefault}{\color[rgb]{0,0,0}$x$}%
}}}}
\put(6751,1199){\makebox(0,0)[lb]{\smash{{\SetFigFont{12}{14.4}{\rmdefault}{\mddefault}{\updefault}{\color[rgb]{0,0,0}$w$}%
}}}}
\put(2521,-241){\makebox(0,0)[lb]{\smash{{\SetFigFont{12}{14.4}{\rmdefault}{\mddefault}{\updefault}{\color[rgb]{0,0,0}$v_C$}%
}}}}
\put(1711,-1051){\makebox(0,0)[lb]{\smash{{\SetFigFont{12}{14.4}{\rmdefault}{\mddefault}{\updefault}{\color[rgb]{0,0,0}a)}%
}}}}
\put(6301,-241){\makebox(0,0)[lb]{\smash{{\SetFigFont{12}{14.4}{\rmdefault}{\mddefault}{\updefault}{\color[rgb]{0,0,0}$y^*$}%
}}}}
\put(9721,-241){\makebox(0,0)[lb]{\smash{{\SetFigFont{12}{14.4}{\rmdefault}{\mddefault}{\updefault}{\color[rgb]{0,0,0}$y$}%
}}}}
\put(7666,-856){\makebox(0,0)[lb]{\smash{{\SetFigFont{12}{14.4}{\rmdefault}{\mddefault}{\updefault}{\color[rgb]{0,0,0}$z$}%
}}}}
\put(1711,-5371){\makebox(0,0)[lb]{\smash{{\SetFigFont{12}{14.4}{\rmdefault}{\mddefault}{\updefault}{\color[rgb]{0,0,0}b)}%
}}}}
\put(2521,-8326){\makebox(0,0)[lb]{\smash{{\SetFigFont{12}{14.4}{\rmdefault}{\mddefault}{\updefault}{\color[rgb]{0,0,0}$v_C$}%
}}}}
\put(1711,-9316){\makebox(0,0)[lb]{\smash{{\SetFigFont{12}{14.4}{\rmdefault}{\mddefault}{\updefault}{\color[rgb]{0,0,0}c)}%
}}}}
\put(6031,-10396){\makebox(0,0)[lb]{\smash{{\SetFigFont{12}{14.4}{\rmdefault}{\mddefault}{\updefault}{\color[rgb]{0,0,0}$x$}%
}}}}
\put(7381,-7516){\makebox(0,0)[lb]{\smash{{\SetFigFont{12}{14.4}{\rmdefault}{\mddefault}{\updefault}{\color[rgb]{0,0,0}$R_1$}%
}}}}
\put(8101,-8371){\makebox(0,0)[lb]{\smash{{\SetFigFont{12}{14.4}{\rmdefault}{\mddefault}{\updefault}{\color[rgb]{0,0,0}$u$}%
}}}}
\put(9721,-4471){\makebox(0,0)[lb]{\smash{{\SetFigFont{12}{14.4}{\rmdefault}{\mddefault}{\updefault}{\color[rgb]{0,0,0}$y$}%
}}}}
\put(9721,-8326){\makebox(0,0)[lb]{\smash{{\SetFigFont{12}{14.4}{\rmdefault}{\mddefault}{\updefault}{\color[rgb]{0,0,0}$y$}%
}}}}
\put(6301,-8326){\makebox(0,0)[lb]{\smash{{\SetFigFont{12}{14.4}{\rmdefault}{\mddefault}{\updefault}{\color[rgb]{0,0,0}$y^*$}%
}}}}
\put(3961,-6316){\makebox(0,0)[lb]{\smash{{\SetFigFont{12}{14.4}{\rmdefault}{\mddefault}{\updefault}{\color[rgb]{0,0,0}$R_2$}%
}}}}
\put(8686,1019){\makebox(0,0)[lb]{\smash{{\SetFigFont{12}{14.4}{\rmdefault}{\mddefault}{\updefault}{\color[rgb]{0,0,0}$Q$}%
}}}}
\put(7561,-3841){\makebox(0,0)[lb]{\smash{{\SetFigFont{12}{14.4}{\rmdefault}{\mddefault}{\updefault}{\color[rgb]{0,0,0}$R_1$}%
}}}}
\put(4546,-7486){\makebox(0,0)[lb]{\smash{{\SetFigFont{12}{14.4}{\rmdefault}{\mddefault}{\updefault}{\color[rgb]{0,0,0}$R_2$}%
}}}}
\put(7381,-9106){\makebox(0,0)[lb]{\smash{{\SetFigFont{12}{14.4}{\rmdefault}{\mddefault}{\updefault}{\color[rgb]{0,0,0}$S$}%
}}}}
\put(7381,-5191){\makebox(0,0)[lb]{\smash{{\SetFigFont{12}{14.4}{\rmdefault}{\mddefault}{\updefault}{\color[rgb]{0,0,0}$S$}%
}}}}
\put(7246, 29){\makebox(0,0)[lb]{\smash{{\SetFigFont{12}{14.4}{\rmdefault}{\mddefault}{\updefault}{\color[rgb]{0,0,0}$S$}%
}}}}
\put(2521,-4471){\makebox(0,0)[lb]{\smash{{\SetFigFont{12}{14.4}{\rmdefault}{\mddefault}{\updefault}{\color[rgb]{0,0,0}$v_C$}%
}}}}
\end{picture}%